\documentclass[conference]{IEEEtran}
\IEEEoverridecommandlockouts
\usepackage{url}
\usepackage{hyperref}
\usepackage{amsthm}
\usepackage{subcaption}
\newtheorem{proposition}{Proposition}
\theoremstyle{remark}
\newtheorem{remark}{Remark}
\usepackage{algorithm}
\usepackage{algpseudocode}
\usepackage{cite}
\usepackage{amsmath,amssymb,amsfonts}
\usepackage{graphicx}
\usepackage{textcomp}
\usepackage{xcolor}
\usepackage{braket}

\def\BibTeX{{\rm B\kern-.05em{\sc i\kern-.025em b}\kern-.08em
    T\kern-.1667em\lower.7ex\hbox{E}\kern-.125emX}}

\renewcommand{\baselinestretch}{0.94}

\begin{document}

\title{
%A Variational Quantum Approach to Finite-Horizon Optimal Control via Hardware-Efficient Ansatz
A Variational Surrogate Approach to Finite-Horizon Quantum Control via Hardware-Efficient Ansatz
\\
\thanks{The authors acknowledge the support of the Danish e-Infrastructure Consortium (DeiC) and the National Quantum Algorithm Academy (NQAA) through the Postdoctoral Scholarship under the project ``Quantum-Driven Solutions for Multi-Agent Systems and Advanced Computation''. This work was also partially supported by UID/00147- Research Center for Systems and Technologies (SYSTEC) - and the Associate Laboratory Advanced Production and Intelligent Systems (ARISE, 10.54499/LA/P/0112/2020) funded by Fundação para a Ciência e a Tecnologia, I.P./ MCTES through the national funds.}
}

\author{\IEEEauthorblockN{Nahid Binandeh Dehaghani}
\IEEEauthorblockA{\textit{Department of Electronic Systems} \\
\textit{Aalborg University}\\
Aalborg, Denmark \\
nahidbd@es.aau.dk}
\and
\IEEEauthorblockN{Rafal Wisniewski}
\IEEEauthorblockA{\textit{Department of Electronic Systems} \\
\textit{Aalborg University}\\
Aalborg, Denmark  \\
raf@es.aau.dk}
\and
\IEEEauthorblockN{A. Pedro Aguiar}
\IEEEauthorblockA{\textit{Faculty of Engineering} \\
\textit{University of Porto}\\
Porto, Portugal \\
pedro.aguiar@fe.up.pt}
}

\maketitle

\begin{abstract}
We present a variational quantum framework for finite-horizon quantum control based on hardware-efficient ansätze. The objective is to steer a quantum system from a given initial state to a desired target state over a fixed time horizon by minimizing a terminal cost defined in terms of state fidelity. Instead of explicitly synthesizing time-dependent control fields or enforcing Hamiltonian reachability constraints, the proposed method reformulates the control objective as a variational optimization problem in which a hardware-efficient parameterized quantum circuit provides a surrogate parameterization of the terminal evolution. The circuit consists of alternating layers of single-qubit rotations and entangling gates, whose parameters are optimized using classical routines to minimize the terminal infidelity. This formulation avoids reliance on problem-specific or physics-inspired ansätze, providing a flexible and implementation-friendly approach compatible with near-term quantum devices. Numerical experiments on multi-qubit state-transfer benchmarks demonstrate high-fidelity state transfer while highlighting the trade-off between ansatz expressivity, optimization complexity, and scalability with respect to system size and circuit depth.
\end{abstract}

\begin{IEEEkeywords}
Quantum computing, optimal control, quantum circuit.
\end{IEEEkeywords}

\section{Introduction}
Optimal control theory provides a systematic framework for steering dynamical systems toward desired objectives under constraints, with applications spanning engineering, physics, and emerging quantum technologies \cite{bryson1975applied, kirk2004optimal}. In the context of quantum systems, optimal control plays a central role in tasks such as state preparation, quantum gate synthesis, and quantum information transfer \cite{glaser2015training, brif2010control, dehaghaniaccess}. Classical approaches to quantum optimal control, including gradient-based methods and techniques derived from Pontryagin’s Minimum Principle, typically rely on explicit parameterizations of time-dependent control fields and accurate system models \cite{dehaghani2023quantumICSC, dehaghani2023quantum}. However, these methods may become computationally demanding for high-dimensional systems and are not always directly compatible with near-term quantum hardware.

In parallel, variational quantum algorithms (VQAs) have emerged as a promising paradigm for leveraging parameterized quantum circuits in combination with classical optimization to solve complex problems \cite{cerezo2021variational, bharti2022noisy}. A key component of VQAs is the choice of ansatz, which defines the structure of the parameterized circuit. Among the various ansätze proposed in the literature, hardware-efficient ansätze have gained significant attention due to their simplicity, low circuit depth, and compatibility with current noisy intermediate-scale quantum (NISQ) devices \cite{kandala2017hardware, schuld2020circuit}. These circuits are typically composed of alternating layers of single-qubit rotations and entangling gates arranged according to hardware connectivity, making them attractive for practical implementations.

Most existing variational approaches to quantum control adopt physics-inspired ansätze that explicitly mimic the underlying system dynamics, for example through Trotterized evolutions of the system Hamiltonian \cite{magann2019feedback, larocca2022diagnosing, dehaghani2025trotterized}. While such approaches can be highly effective and interpretable, they often require detailed knowledge of the system and may not generalize easily across different platforms. 
In contrast, the use of hardware-efficient ansätze as surrogate parameterizations of finite-horizon quantum control remains relatively unexplored, particularly from a control-theoretic perspective.

In this paper, we present a variational quantum approach to finite-horizon optimal control based on hardware-efficient ansätze. Instead of directly parameterizing time-dependent control inputs, we reformulate the control problem as a variational optimization problem in which a parameterized quantum circuit acts as a variational surrogate for the terminal evolution. 
The objective is to map a given initial state to a desired target state over a prescribed finite horizon by minimizing a terminal cost defined in terms of state fidelity.
This formulation provides a flexible and implementation-friendly alternative to conventional control parameterizations by learning a surrogate terminal evolution through a hardware-efficient quantum circuit, while maintaining a clear connection to the underlying finite-horizon control objective.

The main contributions of this work are:
(i) a control-theoretic variational formulation of finite-horizon quantum control using hardware-efficient ansätze,
(ii) a hardware-efficient surrogate control formulation in which trainable circuit parameters replace explicit time-dependent control fields, distinguishing the proposed approach from Hamiltonian-inspired variational quantum control methods, and
(iii) a numerical investigation of convergence, expressivity, and robustness with respect to circuit depth and random initialization.

\section{Optimal Control Formulation}

In this section, we present the finite-horizon quantum optimal control problem considered in this work. We begin with the continuous-time description of a controlled closed quantum system and then introduce a discrete representation over a finite horizon. This formulation provides the control-theoretic basis for the variational circuit framework developed in the next section.

\paragraph*{Controlled Quantum Dynamics}

Consider an $N$-qubit closed quantum system whose state $|\psi(t)\rangle \in \mathbb{C}^{2^N}$ evolves according to the Schrödinger equation
\begin{equation}
i \frac{d}{dt} |\psi(t)\rangle = H(u(t)) |\psi(t)\rangle,
\label{eq:schrodinger}
\end{equation}
where $H(u(t))$ denotes the controlled Hamiltonian. We assume the standard bilinear form
\begin{equation}
H(u(t)) = H_0 + \sum_{k=1}^{m} u_k(t) H_k,
\label{eq:bilinear_ham}
\end{equation}
where $H_0$ is the drift Hamiltonian, $\{H_k\}_{k=1}^m$ are Hermitian control generators, and $u(t) = (u_1(t),\dots,u_m(t))$ is the control input.
The system is initialized at
$
|\psi(0)\rangle = |\psi_{\mathrm{in}}\rangle,
$
and the objective is to steer the system toward a prescribed target state $|\psi_{\mathrm{tar}}\rangle$ at a terminal time $T>0$.
The admissible controls are assumed to belong to the set
\begin{equation*}
\mathcal{U} =
\left\{
\begin{aligned}
u = (u_1,\dots,u_m) : [0,T] &\to \mathbb{R}^m \\
u_k \text{ measurable},\quad &u_k^{\min} \le u_k(t) \le u_k^{\max}
\end{aligned}
\right\}.
\label{eq:admissible_controls}
\end{equation*}

\paragraph*{Finite-Horizon Optimal Control Problem}
The control objective is to maximize the overlap between the terminal state and the prescribed target state, which is equivalent to minimizing the terminal infidelity
$J(u) = 1 - F(u),$
where
$F(u) = \left| \langle \psi_{\mathrm{tar}} \mid \psi(T;u) \rangle \right|^2$
denotes the terminal fidelity.
The resulting finite-horizon optimal control problem is
\begin{equation}
\begin{aligned}
\min_{u \in \mathcal{U}} \quad & J(u) = 1 - \left| \langle \psi_{\mathrm{tar}} \mid \psi(T;u) \rangle \right|^2 \\
\text{s.t.} \quad
& i \frac{d}{dt} |\psi(t)\rangle = H(u(t)) |\psi(t)\rangle, \\
& |\psi(0)\rangle = |\psi_{\mathrm{in}}\rangle .
\end{aligned}
\label{eq:ocp}
\end{equation}
This formulation describes a standard finite-horizon quantum state-transfer problem. Although additional terms may be included to penalize control effort or implementation cost, in this work we focus on the terminal fidelity objective in order to isolate the role of the variational ansatz.

\paragraph*{Time Discretization and Unitary Representation}
For numerical implementation, the interval $[0,T]$ is partitioned into $N_t$ subintervals of equal duration
$\Delta t = \frac{T}{N_t}.$
The control is approximated as piecewise constant over this grid, so that the Hamiltonian is fixed on each interval. The resulting discrete-time evolution is
\begin{equation}
|\psi_{k+1}\rangle = e^{-i H(u_k)\Delta t} |\psi_k\rangle,
\qquad k = 0,\dots,N_t-1,
\label{eq:discrete_dynamics}
\end{equation}
with $|\psi_0\rangle = |\psi_{\mathrm{in}}\rangle$.
Defining the slice-wise propagator
$
U_k := e^{-i H(u_k)\Delta t},
$
the terminal state is given by
\begin{equation}
|\psi(T)\rangle = U_{N_t-1} U_{N_t-2} \cdots U_0 |\psi_{\mathrm{in}}\rangle.
\label{eq:discrete_propagator}
\end{equation} 

This discrete unitary representation makes it possible to reinterpret the control task as an optimization over a finite-dimensional parameter space. In conventional discretized quantum control, these parameters correspond to the sampled control amplitudes. In the present work, however, we do not optimize the sampled controls directly. Instead, we replace the discretized terminal evolution operator with a trainable hardware-efficient circuit, whose parameters provide a variational surrogate for the terminal evolution while preserving the underlying finite-horizon state-transfer objective.

\section{Variational Quantum Formulation}
Starting from the discrete representation introduced in Section II, the finite-horizon evolution may be viewed as a product of unitary propagators associated with the control inputs. Rather than explicitly constructing these propagators from a specified drift Hamiltonian and control Hamiltonians, we represent the overall terminal evolution by a parameterized unitary map \(U(\theta)\in U(d)\), where \(d=2^N\) and \(U(d)\) denotes the group of \(d\times d\) unitary matrices. The corresponding terminal state at the prescribed horizon is given by \(|\psi(\theta)\rangle = U(\theta)|\psi_{\mathrm{in}}\rangle\). In this way, \(\theta\) provides a finite-dimensional surrogate representation of the control strategy over the prescribed horizon.

The control task is then reformulated as the variational optimization problem
\begin{equation}
\min_{\theta \in \mathbb{R}^p} \; J(\theta)
= 1 - \left| \langle \psi_{\mathrm{tar}} \mid \psi(\theta) \rangle \right|^2,
\end{equation}
so that the original search over admissible control functions is replaced by an optimization over circuit parameters.
To realize the unitary $U(\theta)$, we adopt a hardware-efficient ansatz built from alternating layers of local rotations and nearest-neighbor entangling gates. Specifically, we write
\begin{equation}
U(\theta) = \prod_{\ell=1}^{L} (U_{\mathrm{ent}} \, U_{\mathrm{rot}}(\theta_\ell)),
\end{equation}
where
$
U_{\mathrm{rot}}(\theta_\ell)
= \prod_{j=1}^{N}
R_z^{(j)}\!\left(\theta_{\ell,j}^{(1)}\right)
R_x^{(j)}\!\left(\theta_{\ell,j}^{(2)}\right),
$
$
U_{\mathrm{ent}} = \prod_{j=1}^{N-1} \mathrm{CNOT}_{j,j+1},
$
where $\mathrm{CNOT}_{j,j+1}$ denotes a controlled-NOT gate with control qubit $j$ and target qubit $j+1$.
For this circuit structure, the number of trainable parameters is $p = 2NL$.

Unlike Hamiltonian-inspired constructions, this ansatz does not enforce an explicit correspondence between circuit layers and the underlying physical generators. Instead, it provides a flexible and hardware-compatible parameterization of the terminal evolution. From a control perspective, the trainable variables provide a finite-dimensional surrogate parameterization of the control strategy, while the circuit depth $L$ governs the expressive capacity of the variational representation.

The variational problem is solved through a hybrid quantum-classical optimization loop. Starting from an initial parameter vector $\theta^{(0)}$, the hardware-efficient circuit prepares the state $|\psi(\theta^{(r)})\rangle$ at iteration $r$, from which the terminal fidelity and the objective function are evaluated. A classical optimizer then updates the parameter vector, and the process is repeated until a stopping criterion is satisfied. The resulting workflow is summarized in Algorithm~\ref{alg:hea_control}.

\begin{remark}[Control Interpretation]
The proposed variational formulation does not explicitly reconstruct time-dependent control inputs $u(t)$ or require specifying a drift Hamiltonian $H_0$ and control generators $H_k$. Instead, the hardware-efficient circuit provides a surrogate parameterization of the terminal evolution over a fixed horizon. The optimized parameters $\theta^\star$ therefore define an implicit control strategy that preserves the finite-horizon state-transfer objective without explicitly synthesizing physical control fields. The synthesis of implementable control pulses is beyond the scope of the present work and is left for future research.
\end{remark}

\begin{algorithm}[t]
\caption{\small Hybrid Variational Optimization with Hardware-Efficient Ansatz}
\label{alg:hea_control}
\begin{algorithmic}[1] \small
\State \textbf{Input:} initial state $|\psi_{\mathrm{in}}\rangle$, target state $|\psi_{\mathrm{tar}}\rangle$, circuit depth $L$, maximum iterations $R_{\max}$, tolerance $\varepsilon$
\State Initialize the parameter vector $\theta^{(0)} \in \mathbb{R}^{2NL}$
\For{$r = 0,1,\dots,R_{\max}-1$}
    \State Construct the parameterized quantum circuit implementing $U(\theta^{(r)})$
    \State Prepare the terminal state $|\psi(\theta^{(r)})\rangle = U(\theta^{(r)})|\psi_{\mathrm{in}}\rangle$
    \State Evaluate the terminal fidelity
    $
    F(\theta^{(r)}) = \left|\langle \psi_{\mathrm{tar}} \mid \psi(\theta^{(r)}) \rangle\right|^2
    $
    \State Evaluate the objective function
    $
    J(\theta^{(r)}) = 1-F(\theta^{(r)})
    $
    \If{$J(\theta^{(r)}) < \varepsilon$}
        \State \textbf{break}
    \EndIf
    \State Update $\theta^{(r+1)}$ using the SLSQP optimizer
\EndFor
\State \textbf{Output:} optimized parameters $\theta^\star$ and terminal fidelity $F^\star$
\end{algorithmic}
\end{algorithm}

\begin{proposition}[Basic properties of the variational formulation]
Let the parameter set $\Theta \subset \mathbb{R}^{p}$ be nonempty and compact. Assume that the parameterized circuit unitary $U(\theta)$ depends continuously on $\theta \in \Theta$. Consider the variational objective
$
J(\theta)
=
1-
\left|
\langle \psi_{\mathrm{tar}} | U(\theta) | \psi_{\mathrm{in}} \rangle
\right|^2 .
$
Then:
\begin{itemize}
    \item $0 \leq J(\theta) \leq 1$ for all $\theta \in \Theta$.
    \item The objective $J$ is continuous on $\Theta$.
    \item The optimization problem
    $
    \min_{\theta \in \Theta} J(\theta)
    $
    admits at least one global minimizer.
    \item For every $\theta\in\Theta$, $U(\theta)$ defines a valid unitary evolution.
\end{itemize}
\end{proposition}

\begin{proof}
Since $U(\theta)$ is unitary and $|\psi_{\mathrm{in}}\rangle$ and
$|\psi_{\mathrm{tar}}\rangle$ are normalized,
$0\le F(\theta)\le1$, which immediately yields
$0\le J(\theta)\le1$.
The continuity of $U(\theta)$ together with the continuity of
inner products and the squared modulus implies that
$J$ is continuous on $\Theta$.
Since $\Theta$ is compact, the extreme value theorem guarantees
the existence of a global minimizer.
Finally, $U(\theta)$ is a product of unitary single-qubit rotations
and CNOT gates and is therefore unitary for every $\theta$.
\end{proof}
\begin{proposition}[Surrogate reachability of terminal states]
Let $\ket{\psi_{\mathrm{in}}}$ be a fixed initial state and let
$
\mathcal{U}_\Theta := \{ U(\theta) : \theta \in \Theta \} \subset \mathrm{U}(d)
$
denote the family of unitary operators generated by the variational ansatz.
Then, for every $\theta \in \Theta$, the terminal state
$
\ket{\psi(\theta)} = U(\theta)\ket{\psi_{\mathrm{in}}}
$
is reachable through a finite-horizon unitary evolution, that is, there exists a
time-dependent Hamiltonian $H_\theta(t)$ defined on $[0,T]$ such that
\begin{equation}\label{P2}
U(\theta)
=
\mathcal{T}
\exp\!\left(
 -i \int_0^T H_\theta(t)\,dt
\right).
\end{equation}
\end{proposition}
\begin{proof}
For any $\theta\in\Theta$, the operator $U(\theta)$ is unitary by construction. It is a standard result that every unitary operator can be realized as the time-ordered evolution generated by a time-dependent Hermitian Hamiltonian over a finite time interval. Consequently, there exists a (not necessarily unique) Hermitian operator-valued function $H_\theta(t)$ such that \eqref{P2} holds. Applying this evolution to $|\psi_{\mathrm{in}}\rangle$ yields
$|\psi(\theta)\rangle=U(\theta)|\psi_{\mathrm{in}}\rangle$. Hence, every terminal state generated by the variational circuit admits a finite-horizon Hamiltonian realization.
\end{proof}

\begin{remark}[Convergence and Optimization]
The above results guarantee the existence of an optimal variational parameter vector on compact parameter domains. They do not imply that a classical optimizer will necessarily find the global minimizer, since the variational landscape is generally nonconvex. Nevertheless, the continuity and boundedness of the objective function ensure a well-posed optimization problem with stable numerical optimization. In the numerical experiments, we therefore use multiple random initializations to assess the robustness of the obtained solutions.
\end{remark}

\begin{remark}[Expressivity and Circuit Depth]\label{exp}
The expressive power of the hardware-efficient ansatz is strongly influenced by the circuit depth $L$. Increasing $L$ enhances the expressive capacity of the ansatz and enlarges the variational search space, potentially improving the achievable terminal fidelity. However, it also increases the number of trainable parameters, making the optimization landscape more challenging. This expressivity--optimization trade-off is investigated numerically in Section~IV.
\end{remark}

\begin{remark}[Parameter Count and Circuit Complexity]
For the proposed $N$-qubit hardware-efficient ansatz with $L$ layers, each layer applies one $R_z$ and one $R_x$ rotation to every qubit, followed by a nearest-neighbor CNOT entangling block. Consequently, the total number of trainable parameters is
$p = 2NL$,
while the total number of entangling gates is
$N_{\mathrm{CNOT}} = L(N-1)$.
Thus, increasing either the system size $N$ or the circuit depth $L$ enhances the expressive capacity of the ansatz at the cost of a larger parameter space and increased circuit complexity.
\end{remark}

\begin{remark}[Barren Plateaus and Scalability] \label{barren}
Hardware-efficient ansätze are known to exhibit barren plateau phenomena as the system size and circuit depth increase, leading to vanishing gradients under random parameter initializations \cite{larocca2022diagnosing,mcclean2018barren}. Consequently, optimization becomes increasingly challenging for large-scale variational quantum circuits. Here, this limitation is mitigated by restricting attention to moderate system sizes and circuit depths and by employing multiple random initializations to assess robustness. The observed degradation in convergence and increased sensitivity to initialization for larger systems, reported in Section~IV, are consistent with the expected barren plateau behavior. Addressing scalability beyond this regime may require structured ansätze, problem-informed parameterizations, or layer-wise training strategies.
\end{remark}

\begin{figure*}[ht]
    \centering
    \includegraphics[width=0.99\textwidth]{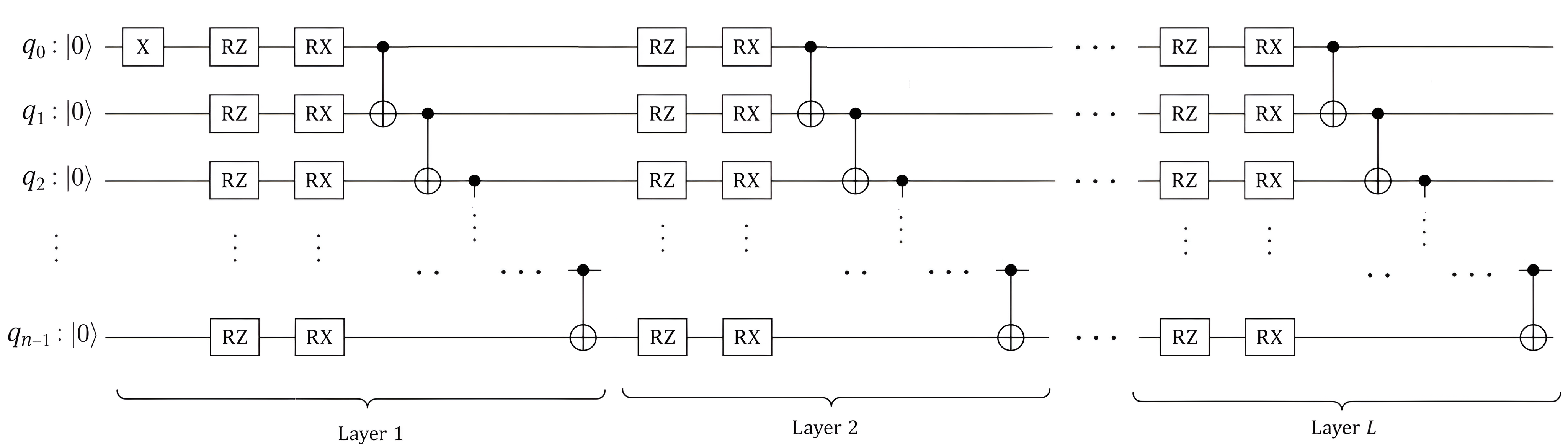}
    \caption{\small Hardware-efficient variational ansatz employed in the proposed framework for an $N$-qubit system. The initial state $|10\cdots0\rangle$ is prepared by applying an $X$ gate to qubit $q_0$. Each variational layer consists of parameterized $R_z$ and $R_x$ rotations on every qubit, followed by a nearest-neighbor CNOT entangling block. Repeating this structure for $L$ layers defines the parameterized unitary $U(\theta)$ used to represent the terminal evolution.}
    \label{fig:ansatz_circuit}
\end{figure*}

\section{Numerical Experiments}

In this section, we evaluate the proposed variational quantum control framework based on a hardware-efficient ansatz. The objective is to assess its ability to achieve high-fidelity state transfer and to investigate how the optimization behavior depends on the system size and circuit depth.

\subsection{Problem Setup}

We consider a finite-horizon quantum state-transfer problem on an $N$-qubit system. The objective is to steer the system from an initial state $|\psi_{\mathrm{in}}\rangle$ to a target state $|\psi_{\mathrm{tar}}\rangle$ using the variational circuit described in Section~III. Specifically, we consider a single-excitation transfer task of the form
$
|\psi_{\mathrm{in}}\rangle = |10\cdots 00\rangle$,
$|\psi_{\mathrm{tar}}\rangle = |00\cdots 01\rangle$.
This benchmark provides a simple yet nontrivial test of the ability of the variational ansatz to learn a state-transfer map across a multi-qubit register. The task consists of transferring a single excitation from one end of the qubit register to the other.
The time horizon is normalized to $T=1$. The variational unitary is realized using the hardware-efficient ansatz introduced in Section~III, with circuit depth denoted by $L$. Each layer consists of parameterized single-qubit rotations followed by a nearest-neighbor CNOT entangling block. The initial excitation is prepared by applying an $X$ gate to the first qubit, as illustrated in Fig.~\ref{fig:ansatz_circuit}. In these experiments, the drift and control Hamiltonians are not explicitly simulated. Instead, the optimized circuit parameters provide a surrogate parameterization of the terminal evolution while preserving the finite-horizon state-transfer objective.
The optimization problem is solved using the Sequential Least Squares Programming (SLSQP) algorithm \cite{bonet2023performance}. The circuit parameters are initialized independently by sampling each parameter uniformly from the interval $[0,4\pi]$, and these bounds are maintained throughout the optimization. The maximum number of iterations is set to $100$, and the stopping tolerance is chosen as $10^{-3}$. All simulations are performed using a noiseless state-vector simulator. To account for the variability induced by random initialization, the results reported below are averaged over multiple independent realizations. Since the optimization problem is generally nonconvex, the final solution may depend on the initial parameter vector.

\paragraph*{Convergence Behavior}
We first examine the convergence behavior of the hybrid optimization procedure for different system sizes. In this experiment, the normalized time horizon is fixed to $T=1$ and the circuit depth is set to $L=10$. We consider $N\in\{3,5,7\}$ qubits, and for each system size the optimization is repeated from five independent random initializations.
Figure~\ref{fig:loss_vs_function_eval} shows the mean loss during optimization, with the shaded region indicating one standard deviation across the five random initializations. For all system sizes, the loss decreases steadily during optimization, indicating that the proposed variational formulation successfully optimizes the terminal state-transfer objective. The final mean fidelities are approximately $0.9996$, $0.9960$, and $0.9554$ for $N=3$, $N=5$, and $N=7$, respectively.
As the number of qubits increases, the optimization becomes more challenging, as reflected by slower convergence and lower final fidelities. This behavior is consistent with the growth of the Hilbert-space dimension, which enlarges the variational search space and reduces the typical overlap between a randomly initialized circuit state and a fixed target state. These observations are also consistent with the barren plateau discussion in Remark~\ref{barren}, where larger systems are expected to exhibit increasingly challenging optimization landscapes.

\begin{figure}[t]
    \centering
    \includegraphics[width=\linewidth]{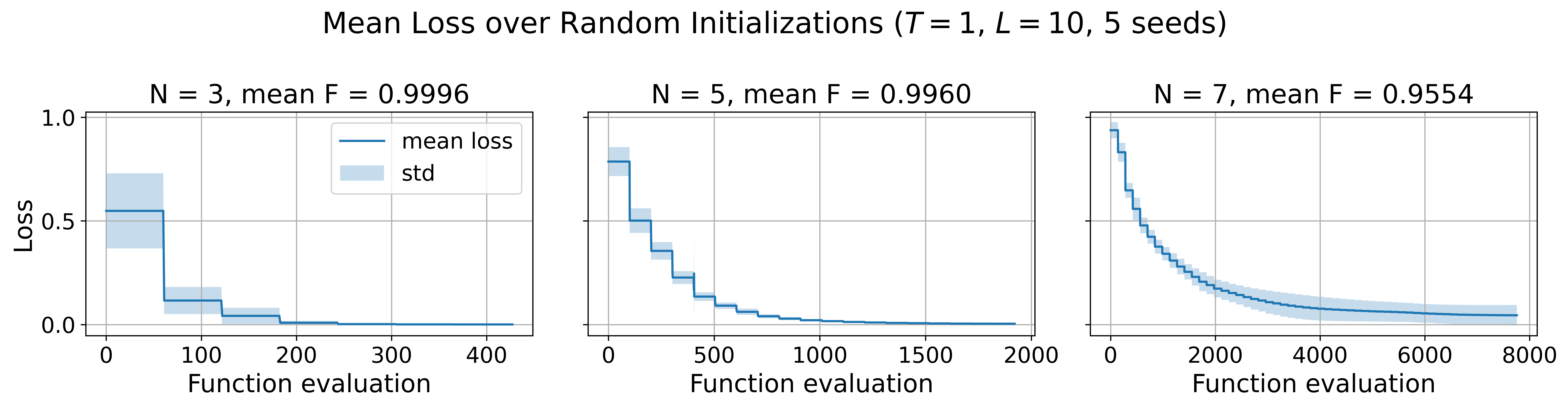}
    \caption{\small Mean loss during optimization for different system sizes $N\in\{3,5,7\}$ with normalized horizon $T=1$ and circuit depth $L=10$. The solid line represents the mean loss over five independent random initializations, and the shaded region indicates one standard deviation.}
    \label{fig:loss_vs_function_eval}
\end{figure}

\paragraph*{Effect of Circuit Depth}
We next investigate the effect of the circuit depth on the final state-transfer fidelity. In this experiment, we consider systems with $N\in\{3,5,7\}$ qubits, representing small- to moderate-scale quantum systems that remain tractable for state-vector simulation while capturing the increasing optimization difficulty associated with larger Hilbert spaces. We consider circuit depths $L\in\{4,6,8,10,12\}$, spanning shallow to relatively deep hardware-efficient circuits to investigate the trade-off between expressive capacity and optimization complexity. Since the circuit depth determines the number of trainable layers in the hardware-efficient ansatz, it directly influences the expressive capacity of the variational unitary.
For each pair $(N,L)$, the optimization is repeated from five independent random initializations, and the mean final fidelity together with one standard deviation is reported. Figure~\ref{fig:depth_fidelity} shows the final fidelity as a function of the circuit depth. For $N=3$, the ansatz achieves consistently high fidelities for all tested depths, with mean fidelities above $0.998$. The performance improves slightly with increasing depth, reaching approximately $0.9997$ at $L=12$. This indicates that relatively shallow circuits are already sufficiently expressive for the three-qubit transfer task.
For $N=5$, the fidelities remain above $0.99$ across all tested depths, although a mild nonmonotonic behavior is observed. The best performance is obtained for $L=12$, with a mean fidelity close to $0.998$. The relatively small standard deviations indicate stable optimization behavior across different random initializations.
For $N=7$, the dependence on circuit depth becomes significantly more pronounced. Shallow and intermediate depths ($L=4$ and $L=6$) lead to substantially lower fidelities and larger variability, reflecting the increased difficulty of the optimization problem for larger Hilbert spaces. In contrast, deeper circuits ($L\geq 8$) recover high-fidelity solutions more consistently, with the best performance achieved at $L=8$, where the mean fidelity exceeds $0.996$. Although high fidelities are still obtained for $L=10$ and $L=12$, the larger error bars indicate increased sensitivity to initialization and a more challenging optimization landscape.
The results demonstrate that increasing the circuit depth generally enhances the expressive capacity of the hardware-efficient ansatz, enabling high-fidelity state transfer for larger systems. However, deeper circuits also increase the number of trainable parameters and may lead to a more challenging optimization landscape. This behavior is consistent with the expressivity--optimization trade-off discussed in Remark~\ref{exp}.

\begin{figure}[t]
    \centering
    \includegraphics[width=\linewidth]{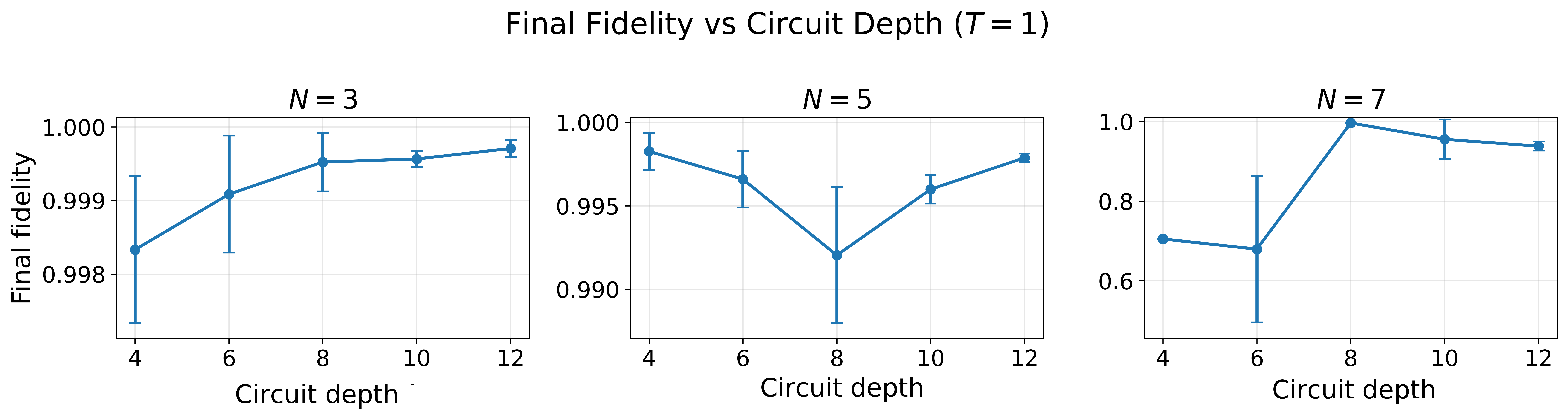}
    \caption{\small Final fidelity as a function of the circuit depth $L$ for system sizes $N=3$, $N=5$, and $N=7$ with normalized horizon $T=1$. The markers represent the mean final fidelity over five independent random initializations, while the error bars indicate one standard deviation.}
    \label{fig:depth_fidelity}
\end{figure}

\paragraph*{Sensitivity to Random Initialization}
We investigate the sensitivity of the hybrid optimization procedure to the initialization of the variational parameters. The circuit depth is fixed to $L=10$, and the optimization is repeated from $20$ independent random initializations for each system size $N\in\{3,5,7\}$.
Figure~\ref{fig:robustness_hist} shows the distribution of the final fidelities obtained across the different runs. For $N=3$, the optimization consistently converges to fidelities very close to unity, with only a small spread across random initializations, indicating a stable optimization landscape for small systems.
For $N=5$, the distribution remains strongly concentrated around high fidelities, although the variability increases slightly. Most runs converge to fidelities above $0.995$, demonstrating that the optimization procedure remains robust for moderate system sizes.
For $N=7$, the distribution becomes significantly broader and exhibits multiple clusters of solutions. While many runs still converge to fidelities close to unity, several terminate at substantially lower fidelities. 
This behavior is consistent with increasingly challenging optimization landscapes and possibly multiple local minima for larger systems.
The proposed variational formulation remains robust for small and moderate system sizes, whereas larger Hilbert spaces introduce increased sensitivity to initialization and greater optimization difficulty. These observations are consistent with the barren plateau discussion in Remark~\ref{barren} and further highlight the increasing optimization complexity as the system size grows.

\begin{figure}[t]
    \centering
    \includegraphics[width=\linewidth]{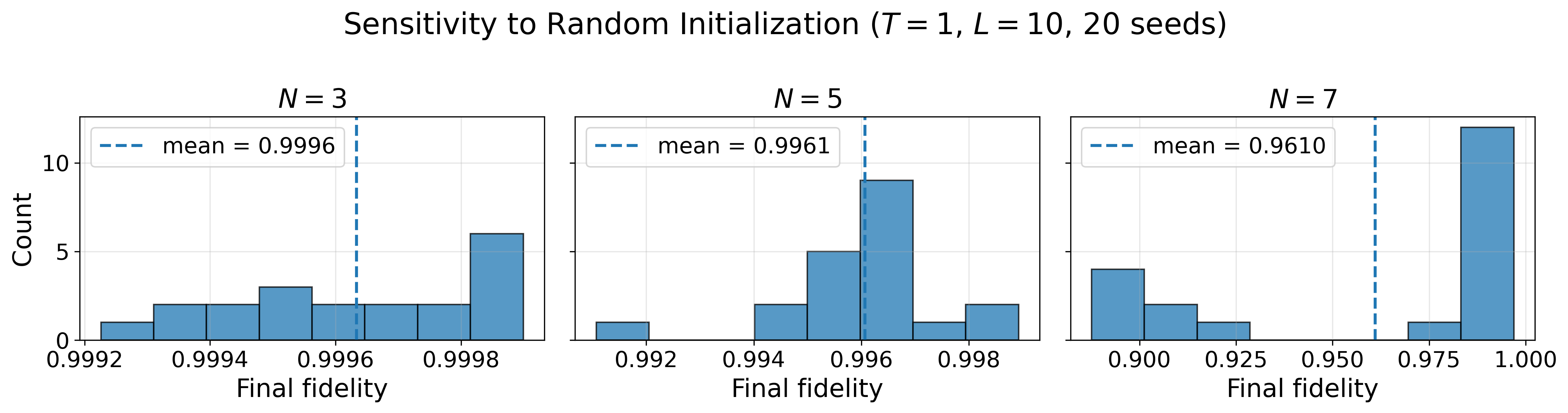} \caption{\small Distribution of the final fidelities obtained from $20$ independent random initializations for system sizes $N=3$, $N=5$, and $N=7$ with fixed circuit depth $L=10$ and normalized horizon $T=1$. The dashed vertical line indicates the mean fidelity for each system size.}
    \label{fig:robustness_hist}
\end{figure}

\paragraph*{Execution of Optimized Circuits on Quantum Backends}
Figure~\ref{fig:before_after_circuit} illustrates representative hardware-efficient circuits before and after variational optimization for the $N=7$, $L=8$ state-transfer task. The circuit topology remains unchanged throughout the optimization process; only the trainable rotation parameters are updated by the classical optimizer. The optimized parameters can subsequently be transferred directly to external quantum simulator backends available through Azure Quantum, demonstrating compatibility with standard gate-based quantum execution workflows.

\begin{figure*}[t]
    \centering
    \begin{subfigure}{0.9\textwidth}
        \centering
        \includegraphics[width=\textwidth]{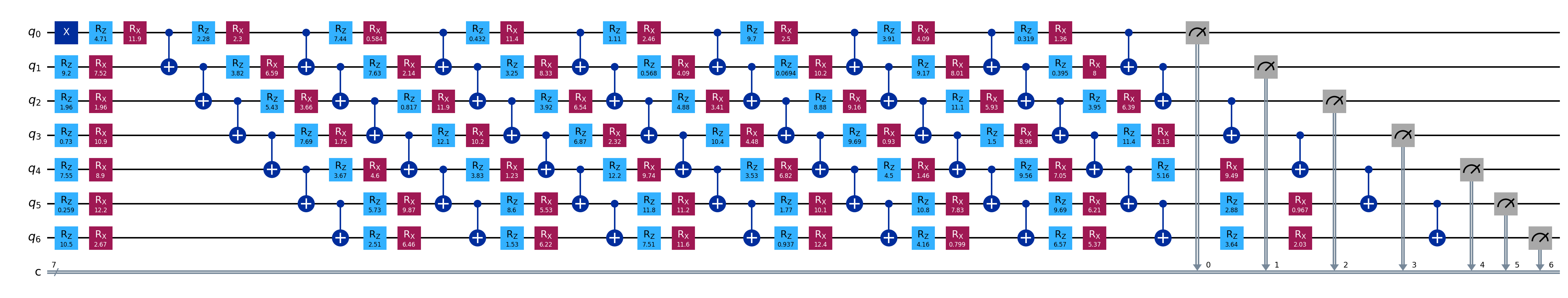}
        \caption{Circuit before optimization with randomly initialized parameters.}
    \end{subfigure}
    \vspace{0.8em}
    \begin{subfigure}{0.95\textwidth}
        \centering
        \includegraphics[width=\textwidth]{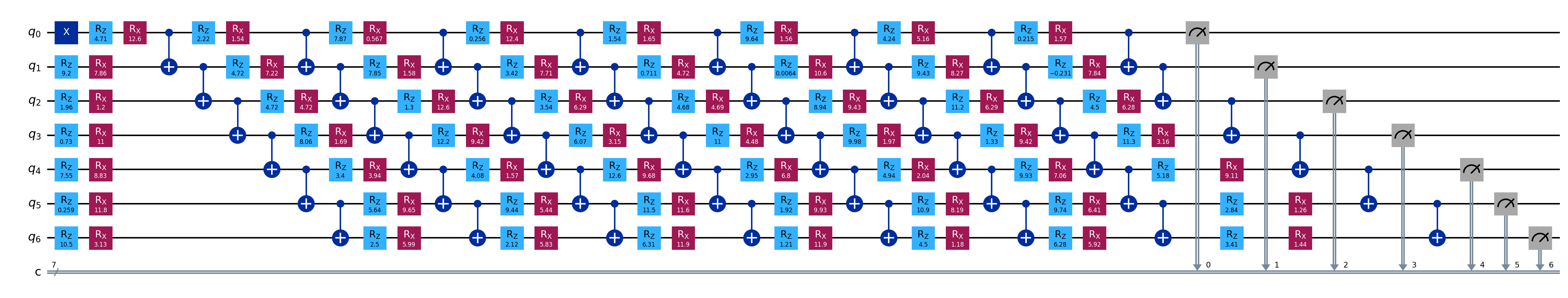}
        \caption{Circuit after variational optimization.}
    \end{subfigure}
\caption{\small Representative seven-qubit hardware-efficient variational circuits with circuit depth $L=8$. The initial state $|10\cdots0\rangle$ is prepared by applying an $X$ gate to qubit $q_0$. The circuit topology remains unchanged before and after optimization; only the trainable rotation angles are updated by the hybrid quantum-classical optimization procedure.}
    \label{fig:before_after_circuit}
\end{figure*}

Table~\ref{tab:backend_results} summarizes the measured target-state probabilities obtained using different Azure Quantum simulator backends with $1000$ shots. Before training, the output probability is distributed over many computational basis states and the target-state probability remains low. After optimization, the probability distribution becomes strongly concentrated on the target bitstring, indicating that the optimized variational parameters transfer consistently across different quantum simulation platforms. In particular, the optimized circuit achieves target-state probabilities above $0.99$ on the IonQ and Rigetti simulators, while the Quantinuum H2 emulator achieves approximately $0.94$. The lower target-state probability observed on the Quantinuum H2 emulator is expected because \texttt{quantinuum.sim.h2-1e} incorporates a realistic hardware and noise model of the H2-1 device rather than representing an ideal noiseless simulator \cite{quantinuum_h2_emulator}.
\begin{table}[h]
\centering
\caption{\small Target-state probabilities before and after optimization for the representative $N=7$, $L=8$ variational circuit executed on different Azure Quantum simulator backends using $1000$ shots.}
\label{tab:backend_results}
\begin{tabular}{lcc}
\hline
Backend & Before training & After training \\
\hline
Quantinuum H2 emulator & $0.033$ & $0.942$ \\
IonQ simulator & $0.038$ & $0.994$ \\
Rigetti QVM simulator & $0.040$ & $0.992$ \\
\hline
\end{tabular}
\end{table}

\section{Conclusion}

We presented a variational quantum framework for finite-horizon state-transfer problems based on a hardware-efficient ansatz. Instead of explicitly optimizing time-dependent control fields, the proposed approach optimizes trainable circuit parameters that provide a surrogate parameterization of the terminal evolution. Numerical results demonstrated high-fidelity state transfer for small and moderate multi-qubit systems, highlighting the role of circuit depth in enhancing expressivity and the increasing sensitivity to initialization for larger systems. The optimized circuits were successfully transferred to multiple Azure Quantum simulator backends, demonstrating compatibility with standard gate-based quantum execution workflows. The proposed framework provides a flexible approach to quantum control that is naturally compatible with near-term quantum computing architectures. The present work considers an open-loop finite-horizon control setting without intermediate measurements or feedback; extending the proposed framework to closed-loop quantum control remains an interesting direction for future work.

\bibliographystyle{IEEEtran}%{IEEEconf}
\bibliography{Main}

@book{bryson1975applied,
  title={Applied optimal control: optimization, estimation and control},
  author={Bryson, Arthur Earl},
  year={2018},
  publisher={Routledge}
}

@book{kirk2004optimal,
  title={Optimal control theory: an introduction},
  author={Kirk, Donald E},
  year={2004},
  publisher={Courier Corporation}
}

@article{glaser2015training,
  title={Training Schr{\"o}dinger’s cat: Quantum optimal control: Strategic report on current status, visions and goals for research in Europe},
  author={Glaser, Steffen J and Boscain, Ugo and Calarco, Tommaso and Koch, Christiane P and K{\"o}ckenberger, Walter and Kosloff, Ronnie and Kuprov, Ilya and Luy, Burkhard and Schirmer, Sophie and Schulte-Herbr{\"u}ggen, Thomas and others},
  journal={The European Physical Journal D},
  volume={69},
  number={12},
  pages={279},
  year={2015},
  publisher={Springer}
}

@article{brif2010control,
  title={Control of quantum phenomena: past, present and future},
  author={Brif, Constantin and Chakrabarti, Raj and Rabitz, Herschel},
  journal={New Journal of Physics},
  volume={12},
  number={7},
  pages={075008},
  year={2010},
  publisher={IOP Publishing}
}

@article{cerezo2021variational,
  title={Variational quantum algorithms},
  author={Cerezo, Marco and Arrasmith, Andrew and Babbush, Ryan and Benjamin, Simon C and Endo, Suguru and Fujii, Keisuke and McClean, Jarrod R and Mitarai, Kosuke and Yuan, Xiao and Cincio, Lukasz and others},
  journal={Nature Reviews Physics},
  volume={3},
  number={9},
  pages={625--644},
  year={2021},
  publisher={Nature Publishing Group UK London}
}

@article{bharti2022noisy,
  title={Noisy intermediate-scale quantum algorithms},
  author={Bharti, Kishor and Cervera-Lierta, Alba and Kyaw, Thi Ha and Haug, Tobias and Alperin-Lea, Sumner and Anand, Abhinav and Degroote, Matthias and Heimonen, Hermanni and Kottmann, Jakob S and Menke, Tim and others},
  journal={Reviews of Modern Physics},
  volume={94},
  number={1},
  pages={015004},
  year={2022},
  publisher={APS}
}

@article{kandala2017hardware,
  title={Hardware-efficient variational quantum eigensolver for small molecules and quantum magnets},
  author={Kandala, Abhinav and Mezzacapo, Antonio and Temme, Kristan and Takita, Maika and Brink, Markus and Chow, Jerry M and Gambetta, Jay M},
  journal={nature},
  volume={549},
  number={7671},
  pages={242--246},
  year={2017},
  publisher={Nature Publishing Group}
}

@article{schuld2020circuit,
  title={Circuit-centric quantum classifiers},
  author={Schuld, Maria and Bocharov, Alex and Svore, Krysta M and Wiebe, Nathan},
  journal={Physical Review A},
  volume={101},
  number={3},
  pages={032308},
  year={2020},
  publisher={APS}
}

@article{magann2019feedback,
  title={Feedback-based quantum optimization},
  author={Magann, Alicia B and Rudinger, Kenneth M and Grace, Matthew D and Sarovar, Mohan},
  journal={Physical Review Letters},
  volume={129},
  number={25},
  pages={250502},
  year={2022},
  publisher={APS}
}

@article{larocca2022diagnosing,
  title={Diagnosing barren plateaus with tools from quantum optimal control},
  author={Larocca, Martin and Czarnik, Piotr and Sharma, Kunal and Muraleedharan, Gopikrishnan and Coles, Patrick J and Cerezo, Marco},
  journal={Quantum},
  volume={6},
  pages={824},
  year={2022},
  publisher={Verein zur F{\"o}rderung des Open Access Publizierens in den Quantenwissenschaften}
}

@inproceedings{dehaghani2023quantumICSC,
  title={Quantum state transfer optimization: Balancing fidelity and energy consumption using Pontryagin maximum principle},
  author={Dehaghani, Nahid Binandeh and Aguiar, A Pedro},
  booktitle={2023 IEEE 11th International Conference on Systems and Control (ICSC)},
  pages={94--99},
  year={2023},
  organization={IEEE}
}

@article{dehaghani2023quantum,
  title={Quantum pontryagin neural networks in gamkrelidze form subjected to the purity of quantum channels},
  author={Dehaghani, Nahid Binandeh and Aguiar, A Pedro and Wisniewski, Rafal},
  journal={IEEE Control Systems Letters},
  volume={7},
  pages={2227--2232},
  year={2023},
  publisher={IEEE}
}

@article{dehaghani2025trotterized,
  title={Trotterized Variational Quantum Control for Spin-Chain State Transfer},
  author={Dehaghani, Nahid Binandeh and Wisniewski, Rafal and Aguiar, A Pedro},
  journal={arXiv preprint arXiv:2511.09684},
  year={2025}
}

@article{mcclean2018barren,
  title={Barren plateaus in quantum neural network training landscapes},
  author={McClean, Jarrod R and Boixo, Sergio and Smelyanskiy, Vadim N and Babbush, Ryan and Neven, Hartmut},
  journal={Nature communications},
  volume={9},
  number={1},
  pages={4812},
  year={2018},
  publisher={Nature Publishing Group UK London}
}

@misc{quantinuum_h2_emulator,
  author       = {{Quantinuum}},
  title        = {{H2 Emulators User Guide}},
  year         = {2025},
  howpublished = {\href{https://docs.quantinuum.com/systems/user_guide/emulator_user_guide/emulators/h2_emulators.html}{\url{https://docs.quantinuum.com/systems/user_guide/emulator_user_guide/emulators/h2_emulators.html}}},
  note         = {Accessed: 2025-08-15}
}

@article{bonet2023performance,
  title={Performance comparison of optimization methods on variational quantum algorithms},
  author={Bonet-Monroig, Xavier and Wang, Hao and Vermetten, Diederick and Senjean, Bruno and Moussa, Charles and B{\"a}ck, Thomas and Dunjko, Vedran and O'Brien, Thomas E},
  journal={Physical Review A},
  volume={107},
  number={3},
  pages={032407},
  year={2023},
  publisher={APS}
}

@article{dehaghaniaccess,
  title={State estimation and control for stochastic quantum dynamics with homodyne measurement: Stabilizing qubits under uncertainty},
  author={Dehaghani, Nahid Binandeh and Aguiar, A Pedro and Wisniewski, Rafal},
  journal={IEEE Access},
  volume={12},
  pages={124729--124739},
  year={2024},
  publisher={IEEE}
}
\end{document}